\documentclass[a4paper,UKenglish]{lipics-v2018-centered}
\usepackage{microtype} 
\nolinenumbers
\newcommand{\citet}{\cite}
\graphicspath{{pics/}{LIPICS/}} 
\usepackage{cite} 
\definecolor{darkgreen}{rgb}{0,0.4,0}
\definecolor{BrickRed}{rgb}{0.65,0.08,0}
\hypersetup{colorlinks=true,linkcolor=blue,citecolor=darkgreen,filecolor=BrickRed,urlcolor=darkgreen}
\newcommand{\LandauO}{O}
\newcommand{\Landauo}{o}
\newcommand{\E}{\mathbb{E}} 
\newcommand{\N}{\mathbb{N}} 
\newcommand{\Dc}{\mathcal{D}} 
\newcommand{\OEIS}[1]{\text{\href{https://oeis.org/#1}{{\small \tt #1}}}} 
\def\GG{\operatorname{GenGamma}}
\def\PGG{\operatorname{ProdGenGamma}}
\def\V{v_{n\ell}}

\usepackage{calc}
\newcommand*{\MyDef}{\, \mathcal{L}\,}
\newcommand*{\eqdefU}{\ensuremath{\mathop{\overset{\MyDef}{=}}}}
\renewcommand*{\=}{\,\mathop{\overset{\MyDef}{\resizebox{\widthof{\eqdefU}}{\heightof{=}}{=}}}\,} 
\def\inlaw{\quad \stackrel{\scriptstyle \mathcal L}{\longrightarrow} \quad } 

\author{Cyril Banderier}{Universit\'e Paris 13, LIPN, UMR CNRS 7030, \url{http://lipn.univ-paris13.fr/~banderier}}{}{https://orcid.org/0000-0003-0755-3022}{} 
\author{Philippe Marchal}{Universit\'e Paris 13, LAGA, UMR CNRS 7539, \url{https://math.univ-paris13.fr/~marchal}}{}{https://orcid.org/0000-0001-8236-5713}{}
\author{Michael Wallner}{Universit\'e de Bordeaux, LaBRI, UMR CNRS 5800, \url{http://dmg.tuwien.ac.at/mwallner}}{}{https://orcid.org/0000-0001-8581-449X}{}
\title{Periodic P\'olya urns and an application to Young tableaux}
\authorrunning{Cyril Banderier \& Philippe Marchal \& Michael Wallner}

\funding{This work was initiated during the postdoctoral position 
of Michael Wallner at the University of Paris Nord, in September-December 2017, thanks to a MathStic funding.
The subsequent part of this collaboration was funded by the Erwin Schr{\"o}dinger Fellowship
of the Austrian Science Fund (FWF):~J~4162-N35.}

\Copyright{Cyril Banderier, Philippe Marchal, Michael Wallner}
\subjclass{
Mathematics of computing $\rightarrow$ 
Enumeration,
Mathematics of computing $\rightarrow$ 
Generating functions, 
Mathematics of computing $\rightarrow$  
Distribution functions,
Mathematics of computing $\rightarrow$  
Ordinary differential equations,
Theory of computation $\rightarrow$  Generating random combinatorial structures,
Theory of computation $\rightarrow$  Random walks and Markov chains 
}
\keywords{P\'olya urn, Young tableau, generating functions, analytic combinatorics, pumping moment, D-finite function, hypergeometric function, generalized Gamma distribution, Mittag-Leffler distribution}

\EventEditors{James Allen Fill and Mark Daniel Ward}
\EventNoEds{2}
\EventLongTitle{29th International Conference on Probabilistic, Combinatorial and Asymptotic Methods for the Analysis of Algorithms (AofA 2018)}
\EventShortTitle{AofA 2018}
\EventAcronym{AofA}
\EventYear{2018}
\EventDate{June 25--29, 2018}
\EventLocation{Uppsala, Sweden}
\EventLogo{}
\SeriesVolume{110}
\ArticleNo{11}

\begin{document}
\maketitle

\begin{abstract}
P\'olya urns are urns where at each unit of time a ball is drawn and is replaced with some other balls 
according to its colour.
We introduce a more general model: The replacement rule depends on the colour of the drawn ball {\em and} the value of the time ($\operatorname{mod} p$). 
We discuss some intriguing properties of the differential operators associated to the generating functions encoding the evolution of these urns.
The initial  {\em partial} differential equation 
indeed leads to {\em ordinary} linear differential equations and we prove that the moment generating functions are D-finite.
For a subclass, 
we exhibit a closed form for the corresponding generating functions (giving the exact state of the urns at time~$n$).
When the time goes to infinity,
we show that these {\em periodic P\'olya urns} follow a rich variety of behaviours:
their asymptotic fluctuations are described by a family of distributions, the generalized Gamma distributions, which can also be seen as powers of Gamma distributions. 
En passant, we establish some enumerative links with other combinatorial objects, and we give an application for a new result on the asymptotics of Young tableaux:
This approach allows us to prove 
that the law of the lower right corner in a triangular Young tableau follows asymptotically a product of generalized Gamma distributions.
\end{abstract}

\newpage 
\section{Periodic P\'olya urns}
\label{sec:urn}

\emph{P\'olya urns} were introduced in a simplified version by George P\'olya and his PhD student Florian Eggenberger 
in~\cite{EggenbergerPolya23,EggenbergerPolya28,Polya30}, with applications to disease spreading and conflagrations.
They constitute a powerful model, still widely used: see e.g.{ }Rivest's recent work on auditing elections~\cite{Rivest18}, or the analysis of deanonymization in Bitcoin's peer-to-peer network~\cite{FantiPramod17}.
They are well-studied objects in combinatorial and probabilistic literature~\cite{AthreyaNey72,FlajoletGabarroPekari05,Mahmoud09},
and offer fascinatingly rich links with numerous objects like random recursive trees, $m$-ary search trees, branching random walks~(see 
e.g.~\cite{BagchiPal85,SmytheMahmoud94,Janson04,Janson05,ChauvinMaillerPouyanne15}).
In this paper we introduce a variation which offers new links with another important combinatorial structure: Young tableaux.
We solve the enumeration problem of this new model,
derive the limit law for the evolution of the urn, and give some applications. 

In the \emph{P{\'o}lya urn model}, one starts with an urn with $b_0$ black balls and $w_0$ white balls at time~$0$. 
At every discrete time step one ball is drawn uniformly at random.
After inspecting its colour it is returned to the urn. 
If the ball is black, $a$ black balls and $b$ white balls are added; if the ball is white, $c$ black balls and $d$ white balls are added (where $a,b,c,d \in \N$ are non-negative integers). 
This process can be described by the so-called \emph{replacement matrix}:
\begin{align*}
	M &=
		\begin{pmatrix}
			a & b \\
			c & d
		\end{pmatrix},
		\qquad
		a,b,c,d \in \N.
\end{align*}
\indent We call an urn and its associated replacement matrix \emph{balanced} if $K:= a+b = c+d$. In other words, in every step the same number $K$ of balls is added to the urn.
This results in a deterministic number of balls after $n$ steps: $b_0 + w_0 + Kn$ balls.



Now, we introduce a more general model which has rich combinatorial, probabilistic, and analytic properties. 

\begin{definition}
	A \emph{periodic P\'olya urn} of period~$p$ with replacement matrices~${M_1,M_2,\ldots,M_p}$
is a variant of a P\'olya urn in which the replacement matrix $M_k$ is used at steps $np+k$. 
Such a model is called \emph{balanced} if each of its replacement matrices is balanced. 
\end{definition}

In this article, we illustrate the aforementioned rich properties on the following model
(the 
results for other values of the parameters  are similar to the case we now handle in detail).

\begin{definition}
	\label{def:youngpolyaurn}
	We call a \emph{Young--P\'olya urn} the periodic P\'olya urn of period~$2$ with replacement matrices~$M_1:=
	\begin{pmatrix}
			1 & 0 \\
			0 & 1
		\end{pmatrix}$ for every odd step, and ~$M_2:=
		\begin{pmatrix}
			1 & 1 \\
			0 & 2
		\end{pmatrix}$ for every even step.
\end{definition}

Let us describe the state of the urn after $n$ steps by pairs (number of black balls, number of white balls), 
starting with $b_0=1$ black ball and $w_0=1$ white ball shown in Figure~\ref{fig:youngurnevolution}. In the first step the matrix $M_1$ is used and gives the two states 
$(2,1), \text{ and } (1,2).$
In the second step, matrix $M_2$ is used,
in the third step, matrix $M_1$ is used again, in the fourth step, matrix $M_2$, etc.
Thus, the possible states are
	$(3,2), (2,3)$, and $(1,4)$, at time $2$, and
	$(4,2), (3,3), (2,4)$, and $(1,5)$, at time $3$.

\begin{figure}[!ht]
		\begin{center}	
\begin{flushleft}
\begin{tabular}{@{}ll@{}}
\begin{tabular}{@{}l@{}}	\includegraphics[width=.63\textwidth]{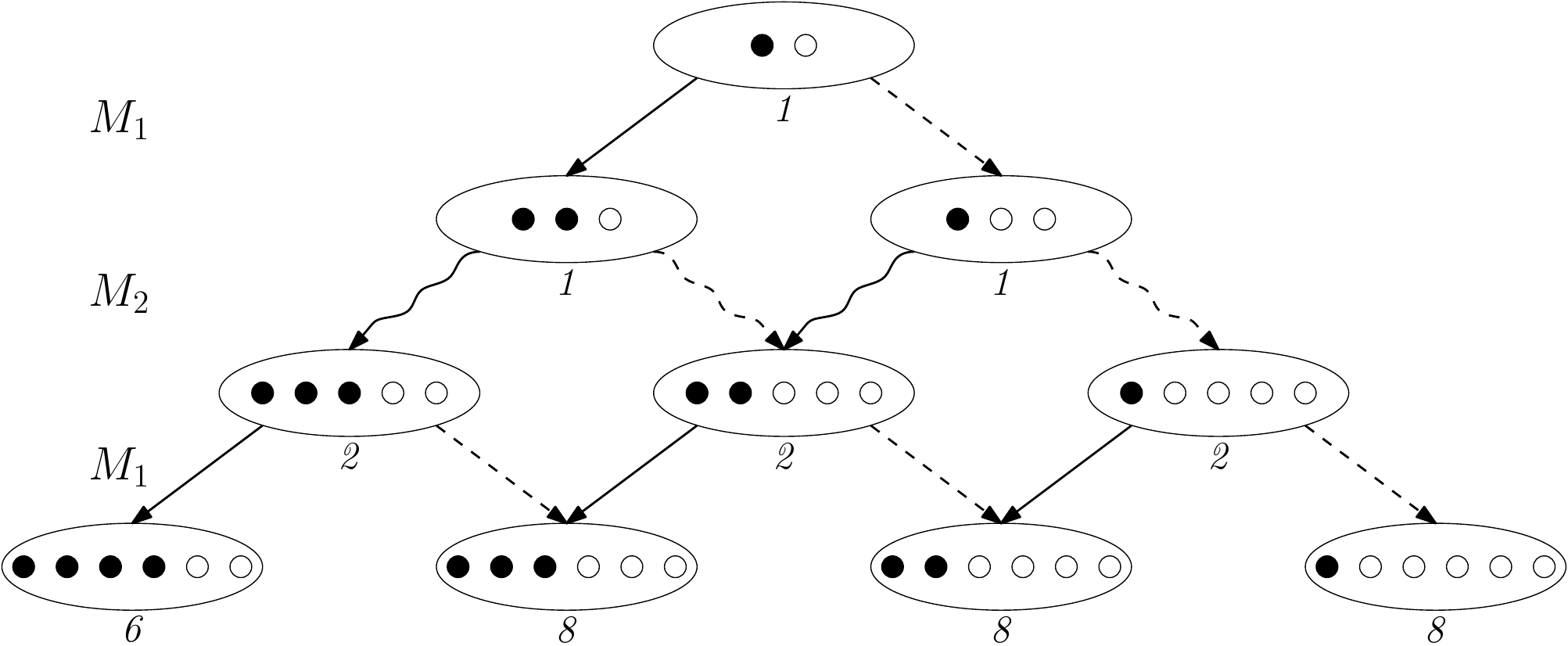} \end{tabular}
&	\hspace{-3.7mm} 		
\scalebox{.91}{\begin{tabular}{@{}l@{}}\\[0mm]
		 $H_0=xy$\\[7mm]
 $H_1=x^2y+xy^2$\\[7mm]
 $H_2=2x^3y^2+2x^2y^3+2xy^4$\\[7mm]
 $H_3=6x^4y^2+8x^3y^3+8x^2y^4+8xy^5$ \\[1cm]
			\end{tabular}} 
			\end{tabular}
			\end{flushleft}
			\caption{The evolution of a Young--P\'olya urn with one initial black and one initial white ball. Black arrows mark that a black ball was drawn, dashed arrows mark that a white ball was drawn. 
			Straight arrows indicate that the replacement matrix $M_1$ was used, curly arrows show that the replacement matrix $M_2$ was used.
			The number below each node is the number of possible transitions to reach such a state. 
			In this article we give a formula 
			for $H_n$ 
 (which encodes all the possible states of the urn at time $n$) 
			and their asymptotic behaviour.}
			\label{fig:youngurnevolution}
		\end{center}
\end{figure}

In fact, each of these states may be reached in different ways, and such a sequence of transitions is called a \emph{history}.
Each history comes with weight one. Implicitly, they induce a probability measure on the states at step $n$.
So, let $B_n$ and $W_n$ be random variables for the number of black and white balls after $n$ steps, respectively. 
As our model is balanced, $B_n + W_n$ is a deterministic process, reflecting the identity $B_n + W_n = b_0 + w_0 + n + \left\lfloor\frac{n}{2}\right\rfloor$. 
So, from now on, we concentrate our analysis on $B_n$.

For the classical model of a single balanced P{\'o}lya urn, the limit law of the random variable $B_n$ is fully known:
The possible limit laws include a rich variety of distributions. 
To name a few, let us mention 
the uniform distribution~\cite{FlajoletDumasPuyhaubert06}, 
the normal distribution~\cite{BagchiPal85}, and the Beta and Mittag-Leffler distributions~\cite{Janson04}. 
Periodic P\'olya urns (which include the classical model) lead to an even larger variety of distributions 
involving a product of \emph{generalized Gamma distributions}~\cite{Stacy62}. 

\def\xxxr{r} 

\begin{definition}\label{def:GG} 
	The generalized Gamma distribution $\GG(\alpha,\beta)$ with real parameters $\alpha,\beta>0$ is defined by the density function (having support $(0,+\infty)$\/) 
	\begin{align*}
		f(x; \alpha,\beta) &:= \frac{ \beta \, x^{\alpha-1} \exp(-x^\beta)}{\Gamma\left(\alpha/\beta\right)},
	\end{align*}
	where $\Gamma$ is the classical Gamma function $\Gamma(z):= \int_0^\infty t^{z-1}\exp(-t) \, dt$.
\end{definition}

\begin{remark}	
	Let ${\mathbf \Gamma}(\alpha)$ be the Gamma distribution\footnote{Caveat: It is traditional to use the same letter for both the $\Gamma$ function and the $\mathbf \Gamma$ distribution.
Also, some authors add a second parameter to the distribution $\mathbf \Gamma$, which is set to $1$ here.} of parameter $\alpha>0$, 
given by its density
	\begin{align*}	
		g(x; \alpha) &= \frac{x^{\alpha-1} \exp(-x)}{\Gamma(\alpha)}.
	\end{align*}
Then, one has ${\mathbf \Gamma}(\alpha)\= \GG(\alpha,1)$ and, for $\xxxr>0$, the distribution of the $r$-th power of a random variable distributed according to ${\mathbf \Gamma}(\alpha)$ is ${\mathbf \Gamma}(\alpha)^\xxxr\=\GG(\alpha/\xxxr,1/\xxxr)$.
\end{remark}

Our main results are the enumeration result from Theorem~\ref{theo:Dfinite}, 
the application to Young tableaux in Theorem~\ref{TheoremCorner}, 
and the following result (and its generalization in Theorem~\ref{theo:PGG}):
\begin{theorem}
	\label{Young_urn}
	The normalized random variable
	$
		\frac{2^{2/3}}{3} \frac{B_n} {n^{2/3}}
	$
	of the number of black balls in a Young--P\'olya urn converges in law to a generalized Gamma distribution:
	\begin{equation*}\frac{2^{2/3}}{3} \frac{B_n}{n^{2/3}} \inlaw \GG\left(1, 3 \right).\end{equation*}
\end{theorem}

We give a proof of this result in Section~\ref{Sec:Moments}. Let us first mention some articles where this distribution has already appeared before:
\begin{itemize} 
	\item in Janson~\citet{Janson10a}, for the analysis of the area of the supremum process of the Brownian motion,
 \item in Pek\"oz, R\"ollin, and Ross~\cite{PekozRollinRoss16}, as distributions of processes on walks, trees, urns, and preferential attachments in graphs
(they also consider what they call a P\'olya urn with immigration, which is a special case of a periodic P\'olya urn),
	\item in Khodabin and Ahmadabadi~\citet{KhodabinAhmadabadi10} following a tradition to 
 generalize special functions by adding parameters in order to capture several probability distributions, such as e.g.~the normal, Rayleigh, and half-normal distribution, as well as the MeijerG function (see also the addendum of~\citet{Janson10a}, mentioning a dozen of other generalizations of special functions).
\end{itemize}


In the next section we translate the evolution process into the language of generating functions by encoding the dynamics of this process into partial differential equations.

\section{A functional equation for periodic P\'olya urns}

Let $h_{n,k,\ell}$ be the number of histories of a periodic P\'olya urn after~$n$ steps with~$k$ black balls and~$\ell$ white balls, with an initial state of $b_0$ black balls and $w_0$ white balls, and with replacement matrices $M_1$ for the odd steps and $M_2$ for the even steps. We define the polynomials
\begin{align*}
	H_{n}(x,y):= \sum_{k,\ell \geq 0} h_{n,k,\ell} x^k y^{\ell}.
\end{align*}
Note that these are indeed polynomials as there are just a finite number of histories after $n$ steps. 
We collect all these histories in the trivariate exponential generating function
\begin{align*}
	H(x,y,z):= \sum_{n \geq 0} H_n(x,y) \frac{z^n}{n!}.
\end{align*}
In particular, we get for the first $3$ terms of $H(x,y,z)$ the expansion (compare Figure~\ref{fig:youngurnevolution})
\begin{align*}
	H(x,y,z) = xy + \left( xy^2 + x^2y \right) z
	 + \left( 2xy^4 + 2x^2 y^3 + 2x^3 y^2 \right) \frac{z^2}{2} 
								+ \ldots
\end{align*}
Observe that the polynomials $H_n(x,y)$ are homogeneous, as we have a balanced urn model.

Now it is our goal to derive a partial differential equation describing the evolution of the periodic P\'olya urn model. 
For a comprehensive introduction to the method we refer to~\cite{FlajoletDumasPuyhaubert06}. 

In order to capture the periodic behaviour we split the generating function $H(x,y,z)$ into odd and even steps. We define
\begin{equation*}
	H_e(x,y,z):= \sum_{n \geq 0} H_{2n}(x,y) \frac{z^{2n}}{(2n)!} \text{\hspace{9mm} and \hspace{9mm}}
	H_o(x,y,z):= \sum_{n \geq 0} H_{2n+1}(x,y) \frac{z^{2n+1}}{(2n+1)!}, 
\end{equation*}
such that $H(x,y,z) = H_e(x,y,z) + H_o(x,y,z)$. 
Next, we associate to the replacement matrices $M_1$ and $M_2$ from Definition~\ref{def:youngpolyaurn} the differential operators $\Dc_1$ and $\Dc_2$, respectively. We get
\begin{equation*}
	\Dc_1:= x^2 \partial_x + y^2 \partial_y \text{\qquad and \qquad}
	\Dc_2:= x^2 y \partial_x + y^3 \partial_y,
\end{equation*}
where $\partial_x$ and $\partial_y$ are defined as the partial derivatives $\frac{\partial}{\partial x}$ and $\frac{\partial}{\partial y}$, respectively. These model the evolution of the urn. For example, in the term $x^2\partial_x$, the derivative $\partial_x$ represents drawing a black ball and the multiplication by $x^2$ returning the black ball and an additional black ball into the urn. 
The other terms have analogous interpretations. 

With these operators we are able to link odd and even steps with the following system
\begin{align}
	\label{eq:funceq1}
	\partial_z H_o(x,y,z) = \Dc_1 H_e(x,y,z) \qquad \text{and} \qquad
	\partial_z H_e(x,y,z) = \Dc_2 H_o(x,y,z).	
\end{align}

Note that the derivative $\partial_z$ models the evolution in time.
This system of partial differential equations naturally corresponds 
to recurrences at the level of coefficients $h_{n,k,\ell}$, and {\em vice versa}.
This philosophy is well explained in the {\em symbolic method} part of~\cite{FlajoletSedgewick09} and see also~\citet{FlajoletDumasPuyhaubert06}.

As a next step we want to eliminate the $y$ variable in these equations. This is possible as the number of balls in each round and the number of black and white balls are connected due to the fact that we are dealing with balanced urns. 
First, as observed previously, one has
\begin{align}
	\label{eq:ballnumber}
\text{number of balls after $n$ steps} =	b_0 + w_0 + n + \left\lfloor \frac{n}{2} \right\rfloor.
\end{align}
Therefore, for any $x^k y^{\ell} z^n$ appearing in $H(x,y,z)$ with $b_0=w_0=1$ we have
\begin{align*}
		k + \ell = 2 + \frac{3n}{2} \text{\quad (if $n$ is even) \quad \qquad and \quad \qquad}
		k + \ell = 2 + \frac{3n}{2} - \frac{1}{2} & \text{ (if $n$ is odd).}	
\end{align*}
This translates directly into
\begin{align}
	\label{eq:funceq2}
	\begin{cases}
	\begin{aligned}
	x\partial_x H_e(x,y,z) + y \partial_y H_e(x,y,z) &= 2 H_e(x,y,z) + \frac{3}{2} z\partial_z H_e(x,y,z), \\[2mm]
	x\partial_x H_o(x,y,z) + y \partial_y H_o(x,y,z) &= \frac{3}{2} H_o(x,y,z) + \frac{3}{2} z\partial_z H_o(x,y,z).
	\end{aligned}
	\end{cases}
\end{align}
Finally, combining \eqref{eq:funceq1} and \eqref{eq:funceq2}, we eliminate 
$\partial_y H_e$ and $\partial_y H_o$. 
After that it is legitimate to insert $y=1$ as there appears no differentiation with respect to $y$ anymore.
As the urns are balanced, the exponents of $y$ and $x$ in each monomial are bound (see Equation~\eqref{eq:ballnumber}),
so we are losing no information on the trivariate generating functions by setting $y=1$.
Hence, from now on we use the notation $H(x,z)$, $H_e(x,z)$, and $H_o(x,z)$ 
instead of $H(x,1,z)$, $H_e(x,1,z)$, and $H_o(x,1,z)$, respectively. 
All of this leads to our first main enumeration theorem:

\begin{theorem}[Linear differential equations and hypergeometric expressions for histories]\label{theo:Dfinite}
The generating functions describing the $2$-periodic Young--P\'olya urn at even and odd time 
 satisfy the following system of differential equations:
\begin{align}
	\label{eq:HeHoPDE}
	\begin{cases}
	\begin{aligned}
	\partial_z H_e(x,z) &= x(x-1) \partial_x H_o(x,z) + \frac{3}{2} z \partial_z H_o(x,z) + \frac{3}{2} H_o(x,z),\\[2mm]
	\partial_z H_o(x,z) &= x(x-1) \partial_x H_e(x,z) + \frac{3}{2} z \partial_z H_e(x,z) + 2 H_e(x,z).
	\end{aligned}
	\end{cases}
\end{align}
\noindent Moreover, $H_e$ and $H_o$ satisfy {\em ordinary} linear differential equations 
(they are {\em D-finite}, see e.g.~\cite[Appendix B.4]{FlajoletSedgewick09} for more on this notion), 
which in return implies that $H=H_e+H_o$ satisfies the equation $L.H(x,z) =0$, where $L$ is a differential operator of order 3 in $\partial_z$,
and one has the hypergeometric closed forms for $h_n:=[z^n] H(1,z)$:
\begin{align}\label{hn}
	h_n &=
	\begin{cases}
		 3^n \frac{\Gamma\left(\frac{n}{2} + 1\right) \Gamma\left(\frac{n}{2} + \frac{2}{3}\right)}{\Gamma(2/3)} & \text{ if $n$ is even,} \\
		3^n \frac{\Gamma\left(\frac{n}{2} + 1/2\right) \Gamma\left(\frac{n}{2} + 7/6\right)}{\Gamma(2/3)} & \text{ if $n$ is odd.}
	\end{cases}
\end{align}

\noindent Alternatively, this sequence satisfies $h(n+2)=\frac{2}{3} h(n+1)+\frac{1}{4}(9\,n^2+21\,n+12)h(n)$.
This sequence is not found in the OEIS\footnote{On-Line Encyclopedia of Integer Sequences, \href{https://oeis.org}{https://oeis.org}.}, 
we added it there, it is now \OEIS{A293653}, and it starts like this:
$1, 2, 6, 30, 180, 1440, 12960, 142560, 1710720, 23950080, 359251200, \ldots$
\end{theorem}

In the next section we will use Equations~\eqref{eq:HeHoPDE} to iteratively derive the moments of the distribution of black balls after $n$ steps.

\section{Moments of periodic P\'olya urns} \label{Sec:Moments}

In this section, we give a proof via the method of moments 
of Theorem~\ref{Young_urn} stated in the introduction.
Let $m_r(n)$ be the $r$-th factorial moment of the distribution of black balls after $n$ steps, i.e.
\begin{align*}
	m_r(n):= \E\left( B_n(B_n-1)\cdots(B_n-r+1) \right).
\end{align*}
Expressing them in terms of the generating function $H(x,z)$, it holds that
\begin{align*}
	m_r(n) = \frac{[z^{n}] \left. \frac{\partial^r}{\partial x^r} H(x,z) \right|_{x=1}}{[z^{n}]H(1,z)}.
\end{align*}
\noindent Splitting them into odd and even moments, we have access to these quantities via the partial differential equation~\eqref{eq:HeHoPDE}.
As a first step we compute $h_n:= [z^n]H(1,z)$, the total number of histories after $n$ steps.
We substitute $x=1$, which makes the equation independent of the derivative with respect to~$x$.
Then, the idea is to transform~\eqref{eq:HeHoPDE} into two independent differential equations for $H_e(1,z)$ and $H_o(1,z)$. This is achieved by differentiating the equations with respect to $z$ and substituting the other one to eliminate $H_e(1,z)$ or $H_o(1,z)$, respectively. This decouples the system, but increases the degree of differentiation by $1$. We get
\begin{align*}
	\left( 9z^2 - 4 \right) \partial_z^2 H_e(1,z) + 39 z \partial_z H_e(1,z) + 24 H_e(1,z) &= 0,\\
	\left( 9z^2 - 4 \right) \partial_z^2 H_o(1,z) + 39 z \partial_z H_o(1,z) + 21 H_o(1,z) &= 0.
\end{align*}
In this case it is easy to extract the underlying recurrence relations and solve them explicitly. 

\noindent This also leads to the closed forms~\eqref{hn} for $h_n$, from which it is easy to compute the asymptotic number of histories for $n \to \infty$. Interestingly, 
the first two terms in the asymptotic expansion
are the same for odd and even number of steps, only the third ones differ. We get
\begin{align*}
	h_n &= n!\frac{\sqrt{\pi}}{2^{1/6}\Gamma\left(\frac{2}{3}\right)} \left(\frac{3}{2}\right)^n n^{1/6} \left(1 + \LandauO\left(\frac{1}{n}\right)\right).
\end{align*}

As a next step we compute the mean. Therefore, we differentiate~\eqref{eq:HeHoPDE} once with respect to~$x$, substitute $x=1$, decouple the system, derive the recurrence relations of the coefficients, and solve them. Note again that the factor $(x-1)$ prevents higher derivatives from appearing and is therefore crucial for this method. After normalization by $h_n$ we get
\begin{align*}
	m_1(n) &=
	\begin{cases}
		\frac{3^{3/2} \Gamma\left(\frac{2}{3}\right)^2}{2\pi} \frac{\Gamma\left(\frac{n}{2} + \frac{4}{3}\right)}{\Gamma\left(\frac{n}{2} + \frac{2}{3}\right)} & \text{ if $n$ is even,}\\
		\frac{3^{3/2} \Gamma\left(\frac{2}{3}\right)^2}{4\pi} \frac{(n+1)\Gamma\left(\frac{n}{2} + \frac{5}{6}\right)}{\Gamma\left(\frac{n}{2} + \frac{7}{6}\right)} & \text{ if $n$ is odd.}
	\end{cases}
\end{align*}
For the asymptotic mean we discover again the same phenomenon that 
the first two terms in the asymptotic expansion are equal for odd and even $n$. 

Differentiating~\eqref{eq:HeHoPDE} to higher orders allows to derive higher moments in a mechanical way
(this however requires further details, which will be included in the expanded version of this article).
In general we get the closed form for the $r$-th factorial moment
\begin{align} \label{mrn}
	m_r(n) &= \frac{3^r}{2^{2r/3}} \frac{\Gamma\left(\frac{r}{3} +\frac{1}{3} \right)}{\Gamma\left(\frac{1}{3}\right)} n^{2r/3} \left(1 + \LandauO\left(\frac{1}{n}\right)\right).
\end{align}

Therefore we see that the moments $\E\left( {B_n^*}^r\right)$ of the rescaled random variable 
$B_n^*:= \frac{2^{2/3}}{3}\frac{B_n}{n^{2/3}}$
converge for $n$ to infinity to the limit
\begin{align} \label{m_r}
	m_r:= \frac{\Gamma\left(\frac{r}{3} +\frac{1}{3} \right)}{\Gamma\left(\frac{1}{3}\right)}.
\end{align}
Note that one has $m_r^{-1/(2r)} = \left(\frac{3e}{r}\right)^{1/6}(1+\Landauo(1))$ for large $r$, so the following sum diverges:
\begin{align}\label{Carleman}
	\sum_{r > 0} m_{r}^{-1/(2r)} = +\infty\,.
\end{align}
Therefore, a result by Carleman (see~\cite[pp.~189-220]{Carleman23} or~\cite[p.~330]{Wall48})\footnote{Note that there is no typo in Formula~\ref{Carleman}:
 if the support of the density is $[0,+\infty[$ the moments in the sum have index~$r$ and exponent $-1/(2r)$, 
while they have index $2r$ and exponent $-1/(2r)$ if the support is  $]-\infty,+\infty[$.}
implies that there exists a unique distribution (let us call it $\mathcal D$) with such moments~$m_r$.

Furthermore, by the asymptotic result from Equation~\eqref{mrn} 
there exist an $n_0>0$ and constants $a_r$ and $b_r$ independent of $n$ such that
$	a_r < m_r(n) < b_r$,
for all $n \geq n_0$. Thus, by the limit theorem of Fr\'echet and Shohat~\cite{FrechetShohat31}\footnote{As a funny coincidence, Fr\'echet and Shohat mention in~\cite{FrechetShohat31} that 
the generalized Gamma distribution with parameter $p\geq 1/2$ is uniquely characterized by its moments.}
 there exists a limit distribution (which therefore has to be~$\mathcal D$)
to which a {\em subsequence} of our rescaled random variables $B_n^*$ converge to.
And as we know via Carleman's criterion above that $\mathcal D$ is uniquely determined by its moments,
it is in fact the {\em full sequence} of $B_n^*$ which converges to $\mathcal D$.

Now it is easy to check that if $X \sim \GG(d,p)$ is a generalized Gamma distributed random variable (as defined in Definition~\ref{def:GG}),
then it is a distribution determined by its moments, which are given by
$\E(X^r) = {\Gamma\left(\frac{d+r}{p}\right)}/{\Gamma\left(\frac{d}{p}\right)}.$

In conclusion, the structure of $m_r$ in Formula~\eqref{m_r} implies that the normalized random variable~$B_n^*$ of the number of black balls in a
Young--P\'olya urn converges to 
$\GG\left(1, 3 \right).$
This completes the proof of Theorem~\ref{Young_urn}. \qed

 \bigskip
The same approach allows us to study the distribution of black balls for the 
urn with replacement matrices $M_1= M_2 = \dots = M_{p-1} = 
	\begin{pmatrix}
			1 & 0 \\
			0 & 1
		\end{pmatrix}$ 
		and 
		$M_p=
		\begin{pmatrix}
			1 & \ell \\
			0 & 1+\ell
		\end{pmatrix}$. 
We call this model the \emph{Young--P\'olya urn of period $p$ and parameter $\ell$}.
		
\begin{theorem} \label{theo:PGG}
The renormalized distribution of black balls in the Young--P\'olya urn of period~$p$ and parameter~$\ell$
is asymptotically a distribution, which we call $\PGG(p,\ell,b_0,w_0)$, defined 
as the following product of independent distributions:
\begin{equation}\label{ProdGenGamma}
\frac{p^\delta}{p+\ell} \frac{B_n }{n^\delta}\inlaw \operatorname{Beta}(b_0,w_0) \prod_{i=0}^{\ell-1} \GG(b_0+w_0+p+i, p+\ell)
\end{equation}
with $\delta=p/(p+\ell)$, and where $\operatorname{Beta}(b_0,w_0)$ is as usual the law with support $[0,1]$ and density $\frac{\Gamma(b_0+w_0)}{\Gamma(b_0)\Gamma(w_0)} x^{b_0-1} (1-x)^{w_0-1}$.
\end{theorem}
\begin{proof}[Sketch]
This follows from the following $r$-th (factorial) moment computation:
\begin{align*}
	\E \left( B_n^r \right) &= \frac{(p+\ell)^r}{p^{\delta r}} \frac{ \Gamma(b_0+r) \Gamma(b_0+w_0)}{\Gamma(b_0) \Gamma(b_0+w_0+r)} \prod_{i=0}^{\ell-1} \frac{\Gamma\left(\frac{b_0+w_0+p+r+i}{p+\ell}\right)}{\Gamma\left(\frac{b_0+w_0+p+i}{p+\ell}\right)} \, n^{\delta r} \left(1+\LandauO\left(\frac{1}{n}\right)\right), 
\end{align*}
which in turn characterizes the $\PGG$ distribution. 
Indeed, if for some independent random variables $X, Y, Z$, one has 
$\E(X^r) = \E(Y^r) \E(Z^r)$ (and if $Y$ and $Z$ are determined by their moments),
then $X \= Y Z$.
\end{proof}

This is consistent with our results on the Young--P\'olya urn introduced in Section~\ref{sec:urn}.
Indeed, there one has $w_0=b_0=1, p=2, \ell=1$, and therefore the renormalized distribution of black balls $\frac{p^\delta}{p+\ell} B_n/n^\delta$ 
is asymptotically $\operatorname{Unif}(0,1) \cdot \GG(4,3) = \GG\left(1, 3 \right)$. 

We will now see what are the implications of this result on an apparently unrelated topic: Young tableaux.

\section{Urns, trees, and Young tableaux}\label{sec:Tableaux}

As predicted by Anatoly Vershik in~\cite{Vershik01}, 
the 21st century should see a lot of challenges and advances on the links of probability theory 
with (algebraic) combinatorics. 
A key role is played here by Young tableaux\footnote{A Young tableau of size~$n$ is an array with columns of (weakly) decreasing height,
in which each cell is labelled, 
and where the labels run from 1 to $n$ and are strictly increasing along rows from left to right and columns from bottom to top, see Figure~\ref{fig:tabtreeurn}. 
We refer to~\cite{Macdonald15} for a thorough discussion on these objects.}, because of their ubiquity in representation theory.
Many results on their asymptotic shape have been collected, 
but very few results are known on their asymptotic content when the shape is fixed
(see e.g.~the works by Pittel and Romik, Angel {\em et al.}, Marchal~\cite{PittelRomik07,AngelHolroydRomikVirag07,Romik15,Marchal15}, who have studied 
the distribution of the values of the cells in random rectangular or staircase Young tableaux, 
while the case of Young tableaux with a more general shape seems to be very intricate).
It is therefore pleasant that our work on periodic P\'olya urns allows us to get advances on the case of a triangular shape, with any slope. 

For any fixed integers $n,\ell, p\geq 1$, we introduce the quantity $N:= p\ell n(n+1)/2$. 
We define a {\em triangular Young tableau} of slope $-\ell/p$ and of size $N$ 
as a classical Young tableau with $N$ cells with length $n\ell$ and height $np$ such that the first $p$ rows (from the bottom) have length $n\ell$, the
next $p$ lines have length $(n-1)\ell$ and so on (see Figure~\ref{fig:tabtreeurn}).
We now study what is the typical value of its lower right corner (with the French convention for drawing Young tableaux, see~\cite{Macdonald15}
but take however care that on page~2 therein, Macdonald advises readers preferring the French convention to ``read this book upside down in a mirror''!).

It could be expected (e.g.~via the Greene--Nijenhuis--Wilf hook walk algorithm for generating Young tableaux, see~\cite{GreeneNijenhuisWilf84}) 
that the entries near the hypotenuse should be $N-o(N)$. 
Can we expect a more precise description of these $o(N)$ fluctuations? Our result on periodic urns enables us to exhibit the right critical exponent, and the limit law in the corner:

\begin{theorem}\label{TheoremCorner}
Choose a uniform random triangular Young tableau $\cal Y$ of slope $-\ell/p$ and 
size~$N = p\ell n(n+1)/2$ and put $\delta=p/(p+\ell)$. 
Let $X_n$ be the entry of the lower right.
Then $(N-X_n)/n^{1+\delta}$ converges in law to the same limiting
distribution as the number of black balls in the periodic Young--P\'olya urn 
with initial conditions $w_0=\ell$, $b_0=p$ and
with replacement matrices $M_1= \dots = M_{p-1} = 
	\begin{pmatrix}
			1 & 0 \\
			0 & 1
		\end{pmatrix}$ 
		and 
		$M_p=
		\begin{pmatrix}
			1 & \ell \\
			0 & 1+\ell
		\end{pmatrix}$,
i.e.~we have the convergence in law, as $n$ goes to infinity:
				\begin{equation*}\frac{p^\delta}{p+\ell} \frac{N-X_n}{ n^{1+\delta} } \stackrel{\mathcal L}{\longrightarrow} \PGG(p,\ell,b_0,w_0).\end{equation*}
(Recall that $\PGG$ is defined by Formula~\ref{ProdGenGamma}.)
\end{theorem}
\noindent {\bf Remark:} The simplest case ($\ell=1$, $p=2$) relates to the Young--P\'olya urn model which we analysed in the previous sections.

\begin{proof}[Sketch of proof.]
We first establish a link between Young tableaux 
and linear extensions of trees. 
Then we will be able to conclude via a link between these trees and 
periodic P\'olya urns. 
Let us start with Figure~\ref{fig:tabtreeurn}, which describes the main characters of this proof.


\begin{figure}[ht!]
		\begin{center}	
			\includegraphics[width=.70\textwidth]{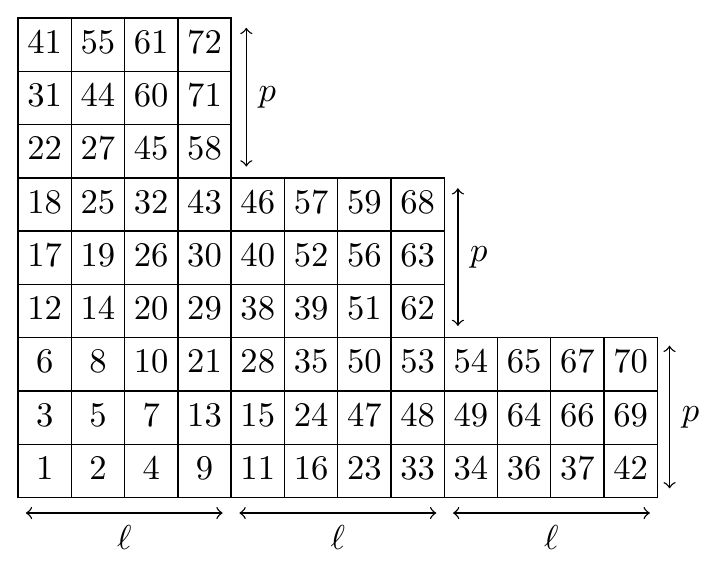} \\ \bigskip \smallskip 
			\includegraphics[width=.95\textwidth]{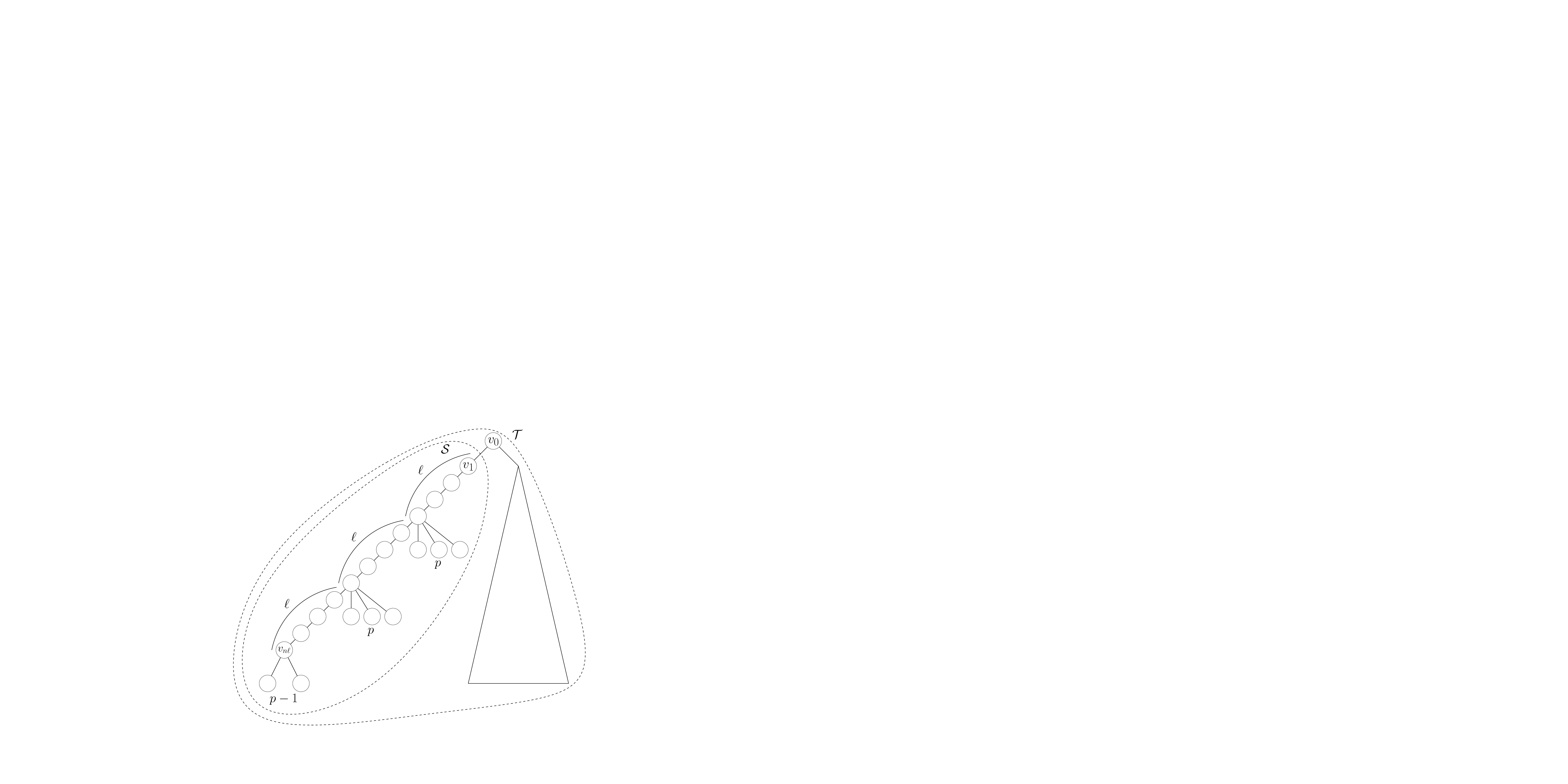}
			\caption{In this section, we see that there is a relation between Young tableaux with a given periodic shape, some trees, and the
			periodic Young--P\'olya urns. The lower right corner of these Young tableaux is thus following the same generalized Gamma distribution we proved for urns.} \label{fig:tabtreeurn}
		\end{center}
\end{figure}

The bottom part of Figure~\ref{fig:tabtreeurn} presents two trees (the ``big'' tree\, $\cal T$, which contains the ``small'' tree $\cal S$).
More precisely, we define the rooted planar tree $\cal S$ as follows:
\begin{itemize}
\item The left-most branch of $\cal S$ has $n\ell+1$ vertices, which we call
$v_1, v_2, \ldots, v_{n\ell+1}$, where $v_1$ is the root and $v_{n\ell+1}$ is the left-most leaf of the tree.
\item For $2\leq k\leq n-1$, the vertex $v_{k\ell}$ has $p+1$ children.
\item The vertex $v_{n\ell}$ has $p-1$ children.
\item All other vertices $v_j$ (for $j<n\ell, j\neq k\ell$) have exactly one child.
\end{itemize}

Now, define $\cal T$ as the ``big'' tree obtained from the ``small'' tree $\cal S$ by adding a vertex $v_0$ as the father of $v_1$ and
adding $N+1-n(p+\ell)$ children to $v_0$ (see Figure~\ref{fig:tabtreeurn}). Remark that the number of vertices of $\cal T$ is equal to 1 + the number of cells of~$\cal Y$. 
Moreover, the hook length of each cell in the first row (from the bottom) of $\cal Y$ is equal to the hook length of the corresponding vertex in the left-most branch of $\cal S$.

Let us now introduce a linear extension $E_{\cal T}$ of~$\cal T$, i.e.~a bijection
from the set of vertices of~$\cal T$ to $\{0, 1, \ldots, N\}$ such that $E_{\cal T}(u)<E_{\cal T}(u')$ whenever $u$ is an ancestor of $u'$.
A key result, which will be proved in the expanded version of this abstract, 
 is the following: if $E_{\cal T}$ is a uniformly random linear extension of $\cal T$, then $X_n$
(the entry of the lower right corner in a uniformly random Young tableau with shape $\cal Y$) has the same law as $E_{\cal T}(\V)$: 
\begin{equation}\label{eq1} X_n \= E_{\cal T}(\V). \end{equation}

What is more, recall that $\cal T$ was obtained from $\cal S$ by adding a root and some children to this root. Therefore, one can obtain
a linear extension of the ``big'' tree\, $\cal T$ from a linear extension of the ``small'' tree $\cal S$ by a simple insertion procedure.
This allows us to construct 
a uniformly random linear extension $E_{\cal T}$ of $\cal T$ and a uniformly random linear extension $E_{\cal S}$ of $\cal S$ such that
\begin{equation*}\left| \frac{2(p+\ell)}{n\ell p}(N-E_{\cal T}(\V))- (n\ell+p-E_{\cal S}(\V)) \right| \rightarrow 0 \text{\qquad (in probability).}\end{equation*}

So, to summarize, we have now 
\begin{equation} \label{eq2} E_{\cal T}(\V) \= E_{\cal S}(\V) + \text{deterministic quantity + smaller order error terms.} \end{equation}

The last step (which we just state here, see our long version for its full proof) is that 
\begin{equation} \label{eq3} E_{\cal S}(\V) \= \text{distribution of periodic P\'olya urn} + \text{deterministic quantity.} \end{equation}
Indeed, more precisely $N-E_{\cal S}(\V)$ has the same law as the number of black balls in a periodic urn
after $(n-1)p$ steps (an urn with period $p$, with adding parameter $\ell$, and with initial conditions 
 $w_0=\ell$ and $b_0=p$). Thus, our results on periodic urns from Section~\ref{Sec:Moments} 
and the conjunction of Equations~\eqref{eq1}, \eqref{eq2}, and \eqref{eq3} gives the convergence in law for $X_n$ which we wanted to prove.
\end{proof}

\section{Conclusion and further work}

In this article, we introduced P\'olya urns with periodic replacements, 
and showed that they can\linebreak be exactly solved with generating function techniques,
and that the initial partial differential equation encoding their dynamics leads to linear (D-finite) moment generating functions,
which we identify as a product of generalized Gamma distributions.
Note that~\cite{Mailler14,KubaSulzbach17} involve the asymptotics of a related process (by grouping $p$ units of time at once of our periodic P\'olya urns). This related process 
is therefore ``smoothing'' the irregularities created by our periodic model, and allows us to connect with the usual famous key quantities for urns, such as the
quotient of eigenvalues of the substitution matrix, etc. Our approach has the advantage to describe each unit of time (and not just what happens after ``averaging'' $p$ units of time at once), 
giving more asymptotic terms, and also exact enumeration. 

In the full version of this work we will consider arbitrary periodic balanced urn models,
and their relationship with Young tableaux.
It remains a challenge to understand the asymptotic landscape of Young tableaux, 
even if it could be globally expected that they behave like a Gaussian free field, 
like for many other random surfaces~\cite{Kenyon01}.
As a first step, understanding the  fluctuations and the universality of the critical exponent at the 
corner could help to get a more global picture. 
Note that our results on the lower right corner directly imply 
similar results on the upper right corner: just use our formulae by exchanging $\ell$ and $p$, 
i.e.~for a slope corresponding to the complementary angle to $90^{\rm o}$.
Thus the critical exponent for the upper right corner is $2-\delta$.
In fact, it is a nice surprise that there is even more structure:
there is a {\em duality} between the limit laws $X$ and $X'$ of these two corners and
we get the factorization as independent random variables (up to renormalization and slight modifications
of the boundary conditions) $XX'\={\mathbf \Gamma}(b_0)$. 
Similar factorizations of the exponential law, which is a particular case of the Gamma distribution, 
have appeared recently in relation with functionals of L\'evy processes, following~\cite{BertoinYor01}. 

\bigskip

{\bf Acknowledgements}: Let us thank C\'ecile Mailler, Henning Sulzbach and Markus Kuba for kind exchanges on their 
work~\cite{Mailler14,KubaSulzbach17} and on related questions.
We also thank our referees for their careful reading.

\bigskip

\bibliographystyle{plainurl} 

\end{document}